\documentclass[submission,copyright,creativecommons]{eptcs}
\providecommand{\event}{GandALF 2021} 
\usepackage{breakurl}             
\usepackage{underscore}           
\usepackage{graphicx}
\usepackage{amsmath}
\usepackage{amssymb}
\usepackage{amsthm}
\usepackage{enumerate}
\usepackage{verbatim}
\usepackage{caption}
\usepackage{mathtools}
\usepackage{microtype}
\usepackage{tikz}

\allowdisplaybreaks

\newcommand*{\Nset}{\mathbb{N}}
\renewcommand{\AA}{A}
\newcommand*{\BB}{B}
\newcommand*{\splits}{\mathrm{Sp}}

\newcommand*{\size}{\mathrm{sz}}

\newcommand*{\FO}{\mathrm{FO}}
\newcommand*{\RE}{\mathrm{RE}}
\newcommand*{\GRE}{\mathrm{GRE}}
\newcommand*{\tower}{\mathrm{twr}}
\newcommand*{\ggame}{\mathrm{GRES}}
\newcommand*{\game}{\mathrm{RES}}
\newcommand*{\sgame}{\mathrm{RESFS}}

\newcommand*{\pair}{\mathtt{set}}

\newcommand*{\enc}{\mathtt{enc}}
\newcommand*{\set}{\varphi}
\newcommand*{\amove}{$a\,$-move}
\newcommand*{\emove}{$\emptyset\,$-move}
\newcommand*{\cupmove}{$\cup\,$-move}
\newcommand*{\catmove}{cat-move}
\newcommand*{\starmove}{$\ast\,$-move}
\newcommand*{\negmove}{$\neg\,$-move}

\newtheorem{theorem}{Theorem}[section]
\newtheorem{lemma}[theorem]{Lemma}

\newtheorem{proposition}[theorem]{Proposition}
\theoremstyle{definition}
\newtheorem{definition}[theorem]{Definition}

\newtheorem{remark}[theorem]{Remark}

\title{Games for Succinctness of Regular Expressions}
\author{Miikka Vilander
\institute{Computing Sciences \\
Tampere University \\
Tampere, Finland}
\email{miikka.vilander@tuni.fi}
}
\def\titlerunning{Games for Succinctness of Regular Expressions}
\def\authorrunning{M. Vilander}
\begin{document}
\maketitle

\begin{abstract}
We present a version of so called formula size games for regular expressions. These games characterize the equivalence of languages up to expressions of a given size. We use the regular expression size game to give a simple proof of a known non-elementary succinctness gap between first-order logic and regular expressions. We also use the game to only count the number of stars in an expression instead of the overall size. For regular expressions this measure trivially gives a hierarchy in terms of expressive power. We obtain such a hierarchy also for what we call RE over star-free expressions, where star-free expressions, that is ones with complement but no stars, are combined using the operations of regular expressions.
\end{abstract}

\section{Introduction}

Even though regular expressions, abbreviated RE, are a very thoroughly studied topic in computer science, little work has been done on their succinctness, or size, until recently. The pioneering paper on the size of RE seems to be in 1974 by Ehrenfeucht and Zeiger \cite{EhrenfeuchtZ1974}. They define the size of an RE as the number of occurrences of alphabet symbols in it and show that there is a deterministic finite automata with $n$ states such that the smallest RE defining the same language has size $2^{n-1}$. In 2005, Ellul et al. \cite{EllulKSW2005} noted the lack of work on succinctness and presented several open problems as well as some results of their own. Some of these open problems were related to the succinctness of RE expanded with operations such as intersection. These and other similar problems were independently solved by Gelade and Neven \cite{Gelade10, GeladeNeven2012} on the one hand and Gruber and Holzer \cite{GruberH08, GruberHolzer2009} on the other. 

Gelade and Neven use a generalization of the result of Ehrenfeucht and Zeiger \cite{EhrenfeuchtZ1974} to obtain double exponential lower bounds for the size of an RE defining the complement of a single RE or the intersection of a finite number of RE in a fixed size alphabet \cite{GeladeNeven2012}. Gelade uses the same technique to also obtain double exponential lower bounds for the added operations of interleaving and counting \cite{Gelade10}. Gruber and Holzer go even further, obtaining tighter bounds for all of the above in a two-letter alphabet \cite{GruberH08, GruberHolzer2009}. They link the size of RE to their star height via a measure on the connectivity of the underlying DFA. The measure is called cycle rank and was first introduced by Eggan and B{\"u}chi \cite{eggan1963}. These two groups worked independently although they were clearly aware of the other group's work.

Many problems in finite model theory have been solved via the use of games such as the famous Ehrenfeucht-Fra{\"i}ss{\'e} game that characterizes quantifier rank or depth in first-order logic. A similar game for RE was presented by Yan \cite{YAN2007181}. This so called split game characterizes the depth of both catenation and stars for generalized regular expressions, or GRE, where complement is added as an operation. Catenation depth is sometimes referred to as dot-depth and star depth is more commonly known as star height. For RE, Hashiguchi famously proved that star height gives a full hierarchy in terms of expressive power \cite{Hashiguchi88}. For GRE, it is notoriously not even known if a language that requires an expression of star height two exists. Yan offers his game as a possible way to attack the generalized star height problem but is only able to complete results on infinite $\omega$-words. 

In the vein of EF-games, there are also games for succinctness. These are often called formula size games. They are games of definability just as the EF-game, but instead of quantifier rank they measure the size of the defining formula. To our knowledge, the earliest example of such a game is for propositional logic by Razborov \cite{Razborov1990}. Perhaps more well known is the later game by Adler and Immerman \cite{AdlerI2003} for a modal logic called $\mathrm{CTL}$. To our knowledge, ours are the first formula size games presented for regular expressions.

While EF-games are played on two structures, formula size games are instead played on two sets of structures, $A$ and $B$. In the context of regular expressions, these sets are languages. Our version of the games also has a resource parameter $k$. The first player S is trying to show that there is an expression $R$ with $A \subseteq L(R)$, $B \subseteq \Sigma^* \setminus L(R)$ and size at most $k$. S essentially sketches the syntax tree of such a separating expression as the game goes on, but in a single game only one branch of the tree is visited. It is the role of the second player D to choose which branch this is, and try to find the error in the strategy of S. A separating expression of appropriate size exists if and only if S has a winning strategy. In addition to the size, in this paper we are also interested in the number of stars in an expression. Thus we add a separate parameter $s$ to the game to track this. The game is very easy to modify in this way to track the number or depth of whatever operators one is interested in.

We use the RE-version of the game to give a simpler proof for a known non-elementary succinctness gap between FO and RE. Stockmeyer \cite{stockmeyerthesis} showed that star-free expressions are non-elementarily more succinct than RE and together with an elementary translation from FO to star-free by McNaughton and Papert \cite{mcnaughton}, the result follows. In addition, we consider the number of stars in an expression as a measure of complexity. For RE a hierarchy in terms of expressive power can be trivially obtained in star height one. For GRE this presents a difficult problem as the full use of complement ramps up the complexity of the game significantly. We present RE over star-free expressions as a natural middle ground between RE and GRE. These include all star-free expressions with complement and their combinations using the operations of RE. For RE over star-free expressions we use a corresponding version of the game to show that the number of stars also gives a full hierarchy in terms of expressive power already in star height one.

The outline of the paper is as follows. In Section 2 we introduce RE, GRE and RE over star-free expressions. We also discuss our definition of size for these expressions and define some notation for the rest of the paper. In Section 3 we present the GRE size game and prove that it works as intended. We also present variations of the game for RE and RE over star-free, and prove some useful lemmas for later. In Section 4 we use the game for RE to show that defining a large finite language requires a large RE. We then define a finite language of non-elementary size via a FO-formula of exponential size, thus reproving the succinctness gap between FO and RE. In Section 5 we show that the number of stars in an expression gives a hierarchy in terms of expressive power for RE over star-free expressions. We conclude in Section~6.
	
\section{Preliminaries}

We begin by defining some basic notions such as regular expressions and our concept of the size of a regular expression. For more on regular expressions we refer the reader to \cite{hopcroft}. We omit the syntax and semantics of first-order logic and direct the reader to \cite{Ebbinghaus1995} for a textbook with a finite model theory approach. 

Let $\Sigma$ be an alphabet. Strings of symbols from the alphabet are called \emph{words} and sets of words are called \emph{languages}. We denote the length of a word $w$ with $|w|$.

The \emph{regular expressions}, or RE, of $\Sigma$ are defined recursively as follows: $\emptyset$, $\epsilon$ and every $a \in \Sigma$ are regular expressions. If $R_1$ and $R_2$ are regular expressions, then also $R_1 \cup R_2$, $R_1R_2$ and $R_1^*$ are regular expressions. The \emph{generalized regular expressions}, or GRE, of $\Sigma$ are defined in the same way with the following addition: if $R$ is a GRE, then $\neg R$ is also a GRE. Sometimes GRE are also defined to include a separate intersection operation. As the effect on succinctness is negligible, we define intersection as the shorthand $R_1 \cap R_2 := \neg (\neg R_1 \cup \neg R_2)$ to keep the number of moves in our game smaller.

The \emph{language of a regular expression} $R$, denoted by $L(R)$ is defined as follows:
\begin{itemize}
	\item $L(\emptyset) = \emptyset$,
	\item $L(\epsilon) = \{\epsilon\}$ (the empty word),
	\item $L(a) = \{a\}$ for $a \in \Sigma$,
	\item $L(R_1 \cup R_2) = L(R_1) \cup L(R_2)$,
	\item $L(R_1R_2) = L(R_1)L(R_2) = \{uv \mid u \in L(R_1), v \in L(R_2)\}$ and
	\item $L(R_1^*) = L(R_1)^* = \{w_1\cdots w_n \mid n \in \Nset, w_i \in L(R_1) \text{ for each } i \in \Nset\}$.
\end{itemize}  
For generalized regular expressions, additionally $L(\neg R_1) = \Sigma^* \setminus L(R_1)$.

We will also refer to \emph{star-free expressions}. These are generalized regular expressions with the $\ast$-rule removed. A classical result by McNaughton and Papert \cite{mcnaughton} states that star-free expressions have the same expressive power over words as first-order logic. Note that this means many languages naturally expressed by a $\RE$ with stars are also expressible by star-free expressions. For example, if $\Sigma = \{a,b\}$, then $L((ab)^*) = L(\epsilon \cup (a\neg\emptyset \cap \neg\emptyset b \cap \neg(\neg\emptyset aa \neg \emptyset) \cap \neg(\neg\emptyset bb \neg \emptyset)))$.

Finally we present a middle ground between RE and GRE we call \emph{RE over star-free expressions}. These expressions are defined by $R$ in the following grammar (we omit parentheses for simplicity):
\begin{align*}
R &::= R \cup R \mid RR \mid R^* \mid S \\
S &::= S \cup S \mid SS \mid \neg S \mid \emptyset \mid \epsilon \mid a \text{ for every }a \in \Sigma
\end{align*}
As the name suggests, RE over star-free expressions include all star-free expressions in the sense of GRE and can combine them using only the operations of RE. Essentially this means that stars cannot occur inside a complement. Since star-free expressions correspond to FO-definable properties of words, we feel this is a natural variation of RE to consider in terms of succinctness. It is quite possible someone else has already presented it but we could not find it in the literature.

There are several ways one could define the size of a regular expression. Gruber and Holzer \cite{GruberH08} use alphabetic width defined as the number of occurrences of symbols from $\Sigma$ in the expression. Gelade and Neven \cite{GeladeNeven2012} on the other hand note that this is not sufficient for GRE since one can construct non-trivial expressions with no symbols from $\Sigma$. Thus they count also operations, ending up with the size of the syntax tree of the expression. This is also sometimes called \emph{reverse polish length} \cite{EllulKSW2005}. We use the latter concept here but the game can easily be adapted to alphabetic width  or actual string length with parentheses if desired.

\begin{definition}
	The \emph{size} of a GRE is defined recursively as follows:
	\begin{itemize}
		\item $\size(\emptyset) = \size(\epsilon) = \size(a) = 1$ for every $a \in \Sigma$,
		\item $\size(R^*) = \size(\neg R) = \size(R) + 1$ and
		\item $\size(R_1 \cup R_2) = \size(R_1R_2) = \size(R_1)+\size(R_2)+1$.
	\end{itemize}
\end{definition}

In the sequel we will deal with some rather large expression sizes. In particular, we will show a \emph{non-elementary} succinctness gap between $\FO$ and $\RE$. This means that the difference in required size is not expressible by an elementary function. In practice, it suffices to show that the size of the $\RE$ is above an exponential tower. For this, we define the function $\tower$ as follows:
\begin{itemize}
\item $\tower(0) = 1$, 
\item $\tower(n+1) = 2^{\tower(n)}$.
\end{itemize}
We also use the shorthand
\[
[n] := \{1, \dots, n\}.
\]

Finally we define some concepts and notations for the $\RE$ size game. First is the concept of regular expressions separating languages.

\begin{definition}
	Let $\AA, \BB \subseteq \Sigma^*$. A $\GRE$ $R$ \emph{separates $\AA$ from $\BB$} if $\AA \subseteq L(R)$ and $\BB \subseteq \Sigma^* \setminus L(R)$. 
\end{definition}

Note that if $\AA = L(R)$ and $\BB = \Sigma^* \setminus L(R)$, then $R$ defines the language $\AA$, so separation is a sort of partial version of defining languages with expressions.
	
To consider catenation and star in the game, we will need notation for the different ways one can split a word into two or more shorter words. 

Let $w \in \Sigma^*$ and $n \in \Nset$. The set of \emph{$n$-splits of $w$} is the set 
\[
\splits^n(w) = \{(w_1, \dots, w_n) \mid w_1 \dots w_n = w\}.
\]
We also use the notation 
\[
\splits(w) := \bigcup\limits_{n \in \Nset} \splits^n(w)
\]
for the set of all splits of $w$.

\section{Generalized regular expression size game}

In this section we define a game for generalized regular expressions that is the equivalent of so called formula size games previously developed for different logics. Since we consider both overall size and number of stars in this paper, we present a game with a separate parameter for stars. 

The GRE size game has two players, Samson (S) and Delilah (D). The game has four parameters: two sets of $\Sigma$-words, $\AA_0$ and $\BB_0$, and two natural numbers $k_0$ and $s_0$ with $k_0 \geq s_0$. Samson wants to show that $\AA_0$ can be separated from $\BB_0$ using a GRE with size at most $k_0$ and at most $s_0$ stars. Delilah wants to refute this. The GRE size game with the above parameters is denoted by $\ggame(k_0, s_0, \AA_0, \BB_0)$.

Positions of the game are of the form $(k, s, \AA, \BB)$ where $\AA$ and $\BB$ are sets of words, $k,s \in \Nset$ and $k \geq s$. The starting position is $(k_0, s_0, \AA_0, \BB_0)$. In a position $P = (k, s, \AA, \BB)$, if $k = 0$, then the game ends and D wins. Otherwise S has a choice of six moves (note that the empty word $\epsilon$ is covered in the \amove):

\begin{itemize}
	\item \amove: S chooses $a \in \Sigma \cup \{\epsilon\}$. If $\AA \subseteq \{a\}$ and $a \notin \BB$, the game ends and S wins. Otherwise D wins.
	\item \emove: If $\AA = \emptyset$, S wins. Otherwise D wins.
	\item \cupmove: S chooses subsets $\AA_1, \AA_2 \subseteq \AA$ such that $\AA_1 \cup \AA_2 = \AA$ and natural numbers $k_1, k_2, s_1, s_2$ such that $k_i \geq s_i$, $k_1 + k_2 + 1 = k$ and $s_1 + s_2 = s$. Then D chooses a number $i \in \{1,2\}$. The game continues from the position $(k_i, s_i, \AA_i, \BB)$.
	\item \catmove: For every $w \in \AA$, S chooses a 2-split $(w_1, w_2)$. Let $\AA_i = \{w_i \mid w \in \AA\}$. Then for every $v \in \BB$, S chooses a function $f_v : \splits^2(v) \to \{1,2\}$. Let $\BB_i = \{v_i \mid f_v(v_1, v_2) = i, (v_1, v_2) \in \splits^2(v)\}$. S chooses numbers $k_1, k_2, s_1, s_2$ such that $k_i \geq s_i$, $k_1 + k_2 + 1 = k$ and $s_1 + s_2 = s$. Finally D chooses a number $i \in \{1,2\}$. The game continues from the position $(k_i, s_i, \AA_i, \BB_i)$.
	\item \starmove: If $\epsilon \in \BB$, D wins. Otherwise, for every $w \in \AA \setminus \{\epsilon\}$, S chooses a natural number $n(w) > 0$ and an $n(w)$-split $(w_1, \dots, w_{n(w)})$ with $w_i \neq \epsilon$ for every $i \in [n(w)]$. Let $\AA' = \{w_i \mid i \in [n(w)], w \in \AA\}$. Then for every $v \in \BB$, S chooses a function $f_v : \splits(v) \to \Nset$ such that $f_v(v_1, \dots, v_n) \in [n]$. Let $\BB' = \{v_i \mid f_v(v_1, \dots, v_n) = i, (v_1, \dots, v_n) \in \splits(v)\}$. The game continues from the position $(k-1, s-1, \AA', \BB')$.
	\item \negmove: The game continues from the position $(k-1, s, \BB, \AA)$.
\end{itemize}

Note that since every move either ends the game or decreases the resource $k$, the game always ends in a finite number of moves and one of the players wins.

We now prove the crucial theorem that states the connection of the game to the succinctness of generalized regular expressions. 

\begin{theorem}\label{thm:gpeli}
	Let $\AA, \BB \subseteq \Sigma^*$ and $k, s \in \Nset$ with $k \geq s$. The following are equivalent:
	\begin{enumerate}
		\item S has a winning strategy in the game $\ggame(k, s, \AA, \BB)$.
		\item There is a generalized regular expression that separates $\AA$ from $\BB$ with size at most $k$ and at most $s$ stars.
	\end{enumerate}
\end{theorem}
\begin{proof}
	In the following we will always have $i \in \{1,2\}$ without explicit statement. We show the equivalence of $1$ and $2$ for all $\AA$ and $\BB$ by induction on the number $k$. The case $k = 0$ is clear. 
	
	$1 \Rightarrow 2$: Let $\delta$ be a winning strategy for S in the game $\ggame(k,\AA, \BB)$. Since $\delta$ is a winning strategy, we have $k > 0$. The proof is divided into cases according to the first move of $\delta$:
	\begin{itemize}
		\item \amove: If the first move is an \amove, because $\delta$ is a winning strategy, we have $\AA \subseteq \{a\} = L(a)$ and $a \notin \BB$ so $\BB \subseteq \Sigma^* \setminus L(a)$. Thus the regular expression $a$ separates $\AA$ from $\BB$.
		\item \emove: Now $\AA = \emptyset$ so $\emptyset$ separates $\AA$ from $\BB$. 
		\item \cupmove: S chooses $\AA_1, \AA_2 \subseteq \AA$ and $k_1, k_2, s_1, s_2$ according to $\delta$. Since $\delta$ is a winning strategy, S has winning strategies from both of the possible following positions $(k_i, s_i, \AA_i, \BB)$. Thus by induction hypothesis there are GREs $R_1$ and $R_2$ such that $R_i$ separates $\AA_i$ from $\BB$,  $\size(R_i) \leq k_i$ and $R_i$ has at most $s_i$ stars. Now $\AA_i \subseteq R_i$ and $\BB \subseteq \Sigma^* \setminus L(R_i)$. Therefore 
		\[
		\AA_0 = \AA_1 \cup \AA_2 \subseteq L(R_1) \cup L(R_2) = L(R_1 \cup R_2).
		\]
		and $\BB \subseteq (\Sigma^* \setminus L(R_1)) \cap (\Sigma^* \setminus L(R_2)) = \Sigma^* \setminus L(R_1 \cup R_2)$ so $R_1 \cup R_2$ separates $\AA$ from $\BB$. In addition, $\size(R_1 \cup R_2) = \size(R_1) + \size(R_2) + 1 \leq k_1 + k_2 + 1 = k$ and $R_1 \cup R_2$ has at most $s_1 + s_2 = s$ stars.
		\item \catmove: S makes his choices according to $\delta$. Now S has a winning strategy for both positions $(k_i, s_i, \AA_i, \BB_i)$ so by induction hypothesis there are GREs $R_1$ and $R_2$ such that $R_i$ separates $\AA_i$ from $\BB_i$, $\size(R_i) \leq k_i$ and $R_i$ has at most $s_i$ stars. Now $\AA_i \subseteq L(R_i)$. For every $w \in \AA$ there are $w_1 \in \AA_1$ and $w_2 \in \AA_2$ such that $w_1w_2 = w$ so $\AA \subseteq L(R_1)L(R_2) = L(R_1R_2)$. On the other side $\BB_i \subseteq \Sigma^* \setminus L(R_i)$. For every $v \in \BB$ and every $(v_1, v_2) \in \splits^2(v)$, either $v_1 \in \BB_1$ or $v_2 \in \BB_2$. Thus $v \notin L(R_1)L(R_2) = L(R_1R_2)$ so $\BB \subseteq \Sigma^* \setminus L(R_1R_2)$. The GRE $R_1R_2$ thus separates $\AA$ from $\BB$. The size and number of stars are handled as in the previous case.
		\item \starmove: S makes his choices according to $\delta$. S has a winning strategy for the following position $(k-1, s-1, \AA', \BB')$ so by induction hypothesis there is a GRE $R$ such that $R$ separates $\AA'$ from $\BB'$, $\size(R) \leq k-1$ and $R$ has at most $s-1$ stars. We have $\AA' \subseteq L(R)$. For every $w \in \AA$ there is $n(w) \in \Nset$ and an $n(w)$-split $(w_1, \dots, w_{n(w)})$ such that $w_j \in \AA'$ for $j \in [n(w)]$. Thus $\AA \subseteq L(R)^* = L(R^*)$. On the other side, $\BB' \subseteq \Sigma^* \setminus L(R)$. For every $v \in \BB$ and every $(v_1, \dots, v_n) \in \splits(v)$, there is $j \in [n]$ such that $v_j \in \BB'$. Thus $v \notin L(R)^* = L(R^*)$ so $\BB \subseteq \Sigma \setminus L(R^*)$. The GRE $R^*$ thus separates $\AA$ from $\BB$. In addition, $\size(R^*) = \size(R) + 1 \leq k$ and $R^*$ has at most $s-1+1 = s$ stars.
		\item \negmove: S has a winning strategy from the following position $(k-1, s, \BB, \AA)$ so there is a GRE $R$ that separates $\BB$ from $\AA$ with $\size(R) \leq k-1$ and at most $s$ stars. Now the GRE $\neg R$ separates $\AA$ from $\BB$. In addition, $\size(\neg R) = \size(R) + 1 \leq k$ and $\neg R$ has at most $s$ stars.
	\end{itemize}
	
	$2 \Rightarrow 1$: Let $R$ be a GRE that separates $\AA$ and $\BB$ with size at most $k$ and at most $s$ stars. The proof is divided into cases according to the outermost operator in $R$:
	\begin{itemize}
		\item $R = a \in \Sigma \cup \{\epsilon\}$: Since $R$ separates $\AA$ from $\BB$, we have $\AA \subseteq \{a\}$ and $\BB \subseteq \Sigma^* \setminus \{a\}$ so $a \notin \BB$. Thus S wins by making an \amove.
		\item $R = \emptyset$: Now $\AA = \emptyset$ so S wins by making a \emove.
		\item $R = R_1 \cup R_2$: Since $R$ separates $\AA$ from $\BB$, we have $\AA \subseteq L(R) = L(R_1) \cup L(R_2)$. Let $\AA_i = \AA \cap L(R_i)$, let $k_1 = \size(R_1)$ and let $k_2 = k - k_1 - 1$. Similarly let $s_1$ be the number of stars in $R_1$ and let $s_2 = s - s_1$. Now $\AA_1 \cup \AA_2 = \AA$, $k_i > s_i$, $k_1 + k_2 + 1 = k$ and $s_1 + s_2 = s$ so these are valid choices for a \cupmove. After the \cupmove, $\AA_i \subseteq L(R_i)$ and $\BB \subseteq \Sigma^* \setminus L(R) = (\Sigma^* \setminus L(R_1)) \cap (\Sigma^* \setminus L(R_2))$ so $\BB \subseteq \Sigma^* \setminus L(R_i)$. Now $R_i$ separates $\AA_i$ from $\BB$. In addition, $\size(R_1) = k_1$,  $\size(R_2) = \size(R)-\size(R_1) -1 \leq k - k_1 - 1 = k_2$. Similarly $R_1$ has $s_1$ stars and $R_2$ has at most $s-s_1 = s_2$ stars. By induction hypothesis, S has a winning strategy for the game $\ggame(k_i, s_i, \AA_i, \BB)$. Together with the first move, this is a winning strategy for the game $\ggame(k, s, \AA, \BB)$.
		\item $R = R_1R_2$: Since $R$ separates $\AA$ from $\BB$, we have $\AA \subseteq L(R) = L(R_1)L(R_2)$. Thus for every $w \in \AA_0$ there is $(w_1, w_2) \in \splits^2(w)$ such that $w_1 \in L(R_1)$ and $w_2 \in L(R_2)$. S makes a \catmove\ and chooses such a split for each $w \in \AA$. On the other side we have $\BB \subseteq \Sigma^* \setminus L(R) = \Sigma^* \setminus L(R_1)L(R_2)$. Thus for every $v \in \BB$ and every $(v_1, v_2) \in \splits^2(v)$, we have $v_1 \notin L(R_1)$ or $v_2 \notin L(R_2)$. For the function $f_v : \splits(v) \to \Nset$, S chooses $i = f_v(v_1, v_2)$ so that $v_i \notin L(R_i)$. S chooses $k_i$ and $s_i$ as in the previous case. Finally we have $\AA_i \subseteq L(R_i)$ and $\BB_i \subseteq \Sigma^* \setminus L(R_i)$ so $R_i$ separates $\AA_i$ from $\BB_i$. The resources $k$ and $s$ are handled like in the previous case. By induction hypothesis, S has a winning strategy from the position $(k_i, s_i, \AA_i, \BB_i)$.
		\item $R = R_1^*$: Since $R$ separates $\AA$ from $\BB$, we have $\AA \subseteq L(R) = L(R_1)^*$. Thus for every $w \in \AA$ there is $(w_1, \dots, w_n) \in \splits(w)$ such that $w_j \in L(R_1)$ for all $j \in [n]$. S makes a \starmove\ and chooses such a split for each $w \in \AA$. On the other side we have $\BB \subseteq \Sigma^* \setminus L(R) = \Sigma^* \setminus L(R_1)^*$. Note that $\epsilon \notin \BB$ so D does not win outright. Now for every $v \in \BB$ and every $(v_1, \dots, v_n) \in \splits(v)$ we have $v_j \notin L(R_1)$ for some $j \in [n]$. For the function $f_v: \splits(v) \to \Nset$, S chooses $j = f_v(v_1, \dots, v_n)$ so that $v_j \notin L(R_1)$. Finally we have $\AA' \subseteq L(R_1)$ and $\BB' \subseteq \Sigma^* \setminus L(R_1)$ so $R_1$ separates $\AA'$ from $\BB'$. In addition, $\size(R_1) = \size(R)-1 \leq k-1$ and $R_1$ has at most $s-1$ stars. By induction hypothesis, S has a winning strategy from the position $(k-1, s-1, \AA', \BB')$.
		\item $R = \neg R_1$: S makes a \negmove. Since $R$ separates $\AA$ from $\BB$, it follows that $R_1$ separates $\BB$ from $\AA$. In addition, $\size(R_1) = \size(R)-1 \leq k-1$ and $R_1$ has at most $s$ stars. By induction hypothesis, S has a winning strategy from the position $(k-1, s, \BB, \AA)$.
	\end{itemize} 
\end{proof}

We have defined the game for generalized regular expressions but this full game turns out to be very complex in a combinatorial sense. For the results in this paper we will use simpler games for RE and RE over star-free.

The $\RE$ size game $\game(k, \AA, \BB)$ is the game $\ggame(k, s, \AA, \BB)$ with the \negmove\ and the star parameter $s$ removed. The proof of Theorem \ref{thm:gpeli} with the \negmove\ cases and $s$ removed proves the following analogue for this game:

\begin{theorem}\label{thm:peli}
	Let $\AA, \BB \subseteq \Sigma^*$, $k \in \Nset$. The following are equivalent:
	\begin{enumerate}
		\item S has a winning strategy in the game $\game(k, \AA, \BB)$.
		\item There is a regular expression that separates $\AA$ from $\BB$ with size at most $k$.
	\end{enumerate}
\end{theorem}

The RE over star-free size game $\sgame(k, s, \AA, \BB)$ is the game $\ggame(k, s, \AA, \BB)$ with the following modification: after a \negmove, the following position is $(k, 0, \BB, \AA)$ instead of the normal $(k, s, \BB, \AA)$. This corresponds with the syntax of RE over star-free, where stars cannot occur under complement. We omit the proof of the analogous theorem for this game:

\begin{theorem}\label{thm:speli}
	Let $\AA, \BB \subseteq \Sigma^*$ and $k, s \in \Nset$ with $k \geq s$. The following are equivalent:
	\begin{enumerate}
		\item S has a winning strategy in the game $\sgame(k, s, \AA, \BB)$.
		\item There is a RE over star-free expression that separates $\AA$ from $\BB$ with size at most $k$ and at most $s$ stars.
	\end{enumerate}
\end{theorem}

As is usual with these sorts of games, we will need a simple lemma stating that if the same word is present on both sides of the game, D has a winning strategy. We prove the lemma for the GRE game and note that it can just as easily be proven for the other variations.

\begin{lemma}\label{lem:sama}
	In a position $P = (k, s, \AA, \BB)$ of a game $\ggame(k_0, s_0, \AA_0, \BB_0)$, if there is $w \in \AA \cap \BB$, then D has a winning strategy from position $P$.
\end{lemma}
\begin{proof}
	Under the assumptions, we describe a strategy for D. For any move of S, this strategy either wins or maintains the condition of having $w \in \AA \cap \BB$. It is thus a winning strategy. We consider the cases for each possible move of S.
	
	\begin{itemize}
		\item \amove: Assume S chooses $a \in \Sigma \cup \{\epsilon\}$. If $\AA \subseteq \{a\}$, then $a = w \in \BB$, so D wins.
		\item \emove: Since $w \in \AA$, $\AA \neq \emptyset$ and D wins.
		\item \cupmove: Assume S chooses subsets $\AA_1, \AA_2 \subseteq \AA$. Since $\AA_1 \cup \AA_2 = \AA$, there is $i \in \{1,2\}$ such that $w \in \AA_i$. D chooses this $i$ and in the following position $(k_i, s_i, \AA_1, \BB)$, we have $w \in \AA_i \cap \BB$.
		\item \catmove: Let $(w_1, w_2)$ be the split S chooses for $w$ on the $\AA$-side and let $f_w : \splits^2(w) \to \{1,2\}$ be the function S chooses for $w$ on the $\BB$-side. D chooses the number $i := f_w(w_1,w_2)$. In the following position $(k_i, s_i, \AA_i, \BB_i)$, we have $w_i \in \AA_i \cap \BB_i$.
		\item \starmove: If $w = \epsilon$, D wins. Otherwise, let $(w_1, \dots, w_n)$ be the split S chooses for $w$ on the $\AA$-side and let $f_w: \splits(w) \to \Nset$ be the function S chooses for $w$ on the $\BB$-side. Let $i := f_w(w_1, \dots, w_n)$. In the following position $(k-1, s-1, \AA', \BB')$ we have $w_i \in \AA' \cap \BB'$.
		\item \negmove: In the following position $(k-1, s, \BB, \AA)$, we have $w \in \BB\cap\AA$.
	\end{itemize}	
\end{proof}

For the RE over star-free game, we need a further lemma that gives an easy condition to guarantee that the current sets $\AA$ and $\BB$ cannot be separated via a star-free expression. The language we use for the game has words with long strings of the same symbol in them. We call these \emph{$a$-chains} for $a \in \Sigma$. For example, the word $baabbaaa$ has two $a$-chains of lengths 2 and 3 respectively. We use the GRE game with $s = 0$ to argue about star-free expressions.

\begin{lemma}\label{lemma:cup}
	In a position $P = (k, 0, A, B)$ of a game $\ggame(k_0, s_0, A_0, B_0)$, if there are $w \in A$ and $w' \in B$ such that they only differ from each other by lengths of one or more chains of symbols, each of length more than $k$ in both, then D has a winning strategy from position $P$.
\end{lemma}
\begin{proof}
	We describe a strategy for D. For each move of S, this strategy either wins or maintains the assumptions of the lemma so it is a winning strategy. We consider each possible move of S:
	\begin{itemize}
		\item \amove: S chooses $a \in \Sigma \cup \epsilon$. Since $w$ has a chain with length more than $k > 0$, clearly $w \neq a$ so D wins.
		\item \emove: Since $w \in A$, $A \neq \emptyset$ and D wins.
		\item \cupmove: S chooses subsets $A_1, A_2 \subseteq A$. Since $A_1 \cup A_2 = A$, we have $w \in A_i$ for some $i \in \{1,2\}$. D chooses this $i$ and in the following position $(k_i, 0, A_i, B)$ we have $w \in A_i$ and $w' \in B$. In addition, the chains  of $w$ and $w'$ that differ are of length more than $k > k_i$. Thus the assumptions still hold.
		\item \catmove: Let $(w_1, w_2)$ be the split S chooses for $w \in A$ and let $f_{w'} : \splits^2(w') \to \{1,2\}$ be the function S chooses for $w' \in B$. Let $k_1, k_2$ be the numbers chosen by S with $k_1 + k_2 + 1 = k$. Since $w$ and $w'$ only differ by the lengths of some chains, for each chain in $w$ we can find the corresponding chain in $w'$. 
		
		If the split $(w_1, w_2)$ splits no chains where $w$ and $w'$ differ, then we consider the split $(w'_1, w'_2)$ of $w'$ at the corresponding point and in the following position $(k_i, 0, A_i, B_i)$, the assumptions hold since $k_i < k$. 
		
		Now assume $(w_1, w_2)$ splits a chain of length more than $k$ and the length of this chain is different but still more than $k$ in $w'$. If the length of the chain in $w_i$ is at more than $k_i$ for both $i$, then we consider a split $(w'_1, w'_2)$ of $w'$ where the same holds. Recall such a split can be found since $k_1 + k_2 + 1 = k$ and the length of the chain is more than $k$ in $w'$ also. Now the assumptions hold in the following position. 
		
		Otherwise, by symmetry we assume that the length of the chain in $w_1$ is less than or equal to $k_1$. In this case we consider the split $(w'_1, w'_2)$ of $w'$ where the length of the chain in $w'_1$ is identical to $w_1$. Now the lengths of the chains in $w_2$ and $w'_2$ are more than $k_2$ since $k_1 + k_2 + 1 = k$. Thus if the following position is $(k_2, 0, A_2, B_2)$, then the assumptions hold. If the following position is $(k_1, 0, A_1, B_1)$, then either there are still other differing chains of length more than $k > k_1$ and the assumptions hold, or $w_1 = w'_1$ and D has a winning strategy by Lemma \ref{lem:sama}.

		\item \starmove: We assume that the star resource $s = 0$ in the position $P$ so S cannot make a \starmove.
		\item \negmove: In the following position $(k-1, 0, B, A)$, the assumptions still hold as they are symmetric w.r.t. $A$ and $B$ and $k-1 < k$.
	\end{itemize}
\end{proof}

\begin{remark}\label{remark}
	The $\GRE$ size game can be modified in several ways to obtain different games. The games for $\RE$ and $\RE$ over star-free are examples of this. Additional operations can be included by adding moves. For example the move corresponding to intersection is the union move with the roles of $\AA$ and $\BB$ switched. One could also have separate resources for different operations or ignore some operations entirely. It is also possible to modify how the resources work with binary moves to track the nesting depth of an operation instead of the number. 
\end{remark}

\section{The succinctness gap between FO and RE}

To compare the succinctness of $\FO$ and $\RE$, we must restrict the models of $\FO$ to \emph{word models}. These are finite models with a linear order and unary predicates to indicate which letter of the alphabet $\Sigma$ is in each spot. Thus properties of words are often defined in a language of the form $\FO(<, P_1, \dots, P_n)$.

In his thesis \cite{stockmeyerthesis} Stockmeyer showed that star-free generalized regular expressions are non-ele\-men\-ta\-ri\-ly more succinct than regular expressions. Since there is an elementary translation from $\FO$ to star-free expressions \cite{mcnaughton}, this implies that $\FO$ is non-elementarily more succinct than $\RE$. The proof of Stockmeyer is quite involved as he encodes computations of Turing machines into star-free expressions. In this section, we show a simple way to obtain the gap between $\FO$ and $\RE$ via the RE size game. Our proof relies on the following proposition which states that to define a large finite language with a $\RE$, the $\RE$ must be quite large as well.

\begin{proposition}\label{thm:nonelem}
	A finite language $L$ cannot be defined via a $\RE$ with size less than $\log|L|$.
\end{proposition}
\begin{proof}
	Let $L$ be a finite language and $k_0 < \log|L|$. We consider the game $\game(k_0, L, \Sigma^* \setminus L)$. We will show that after every move of S, D will either gain a winning strategy via Lemma \ref{lem:sama}, or D can maintain the following two conditions in any position $(k, \AA, \BB)$ of the game:
	\begin{align*}
		&1.\ \  k \leq \log(|\AA|) \\
		&2.\ \  \Sigma^{>N} := \{w \in \Sigma^* \mid |w|>N\} \subseteq \BB \text{ for some } N \in \Nset 
	\end{align*}
	In the starting position $(k_0, L, \Sigma^* \setminus L)$, we have $k_0 \leq \log(|L|)$ so condition 1 holds. For condition 2, note that since $L$ is finite, $\Sigma^* \setminus L$ includes every word with length greater than the maximum length of words in the language $L$. 
	
	Consider a position $(k, \AA, \BB)$ of the game $\game(k_0, L, \Sigma^* \setminus L)$ and assume conditions 1 and 2 hold. S has five different moves to choose from:
	\begin{itemize}
		\item \starmove: Since $0 < k \leq \log(|\AA|)$, we have $|\AA| \geq 2$ so there is $w \in \AA$ with $w \neq \epsilon$. Let $(w_1,w_2, \dots ,w_m)$ be the split chosen by S for $w$. By condition 2, there is $N \in \Nset$ such that $\Sigma^{>N} \subseteq \BB$. Let $v = w_1^{N+1}$. Now $|v| > N$ so $v \in \BB$. For the split $(w_1, w_1, \dots ,w_1)$ of $v$ S must choose the piece $w_1$ so in the following position $(k-1, \AA', \BB')$, we have $w_1 \in \AA' \cap \BB'$ and by Lemma \ref{lem:sama}, D has a winning strategy from this position.
		\item \cupmove: Let $\AA_1, \AA_2 \subseteq \AA$ and $k_1, k_2 < k$ be the choices of S. If either $\AA_i$ is empty, D chooses the other one and both conditions are trivially maintained. Assume both $\AA_i$ are non-empty. Since $\AA_1 \cup \AA_2 = \AA$, we obtain $|\AA_1| + |\AA_2| \geq |\AA|$. Now we have $k_i \leq \log(|\AA_i|)$ for some $i \in \{1,2\}$, since otherwise
		\begin{align*}
		k &= k_1 + k_2 + 1 > \log(|\AA_1|) + \log(|\AA_2|) + 1 \\
		&= \log(|\AA_1||\AA_2|)+1 \geq \log(|\AA_1|+|\AA_2|) \geq \log(|\AA|) \geq k,
		\end{align*}
		which is a contradiction. D chooses such an $i$, fulfilling condition 1 in the following position is $(k_i, \AA_i, \BB)$. Condition 2 is trivially maintained since $\BB$ remains unchanged in $\cup$-moves.
		\item \catmove: Let the two possible following positions be $P_i = (k_i, \AA_i, \BB_i)$ for $i \in \{1,2\}$. We consider condition 2 first. Let $w \in \Sigma^{>N}$. Let $v \in \AA$ and let $(v_1, v_2) = v$ be the split chosen by S for $v$. Now $u = v_1w \in \Sigma^{>N} \subseteq \BB$. For the split $(v_1, w)$ of $u$, if S chooses the piece $v_1$, then $v_1 \in \AA_1 \cap \BB_1$ and by Lemma \ref{lem:sama}, D has a winning strategy from position $P_1$. Thus we assume that S chooses the piece $w$ and $w \in \BB_2$. In the same way using the word $wv_2$, we get $w \in \BB_1$. Thus, in order to not give D a winning strategy via Lemma \ref{lem:sama}, S must maintain condition 2 for both positions $P_i$.
		
		Now let us address condition 1. Since for every $w \in \AA$ there is $w_1 \in \AA_1$ and $w_2 \in \AA_2$ such that $w_1w_2=w$, we obtain $|\AA_1||\AA_2| \geq |\AA|$. We again have $k_i \leq \log(|\AA_i|)$ for some $i \in \{1,2\}$, since otherwise
		\[
		k = k_1 + k_2 + 1 > \log(|\AA_1|) + \log(|\AA_2|) + 1 = \log(|\AA_1||\AA_2|)+1 \geq \log(|\AA|) \geq k,
		\]
		which is a contradiction. D again fulfills condition 1 by choosing such an $i$. 
		\item $a\,$- or \emove: Since $0 < k \leq \log(|\AA|)$, we have $|\AA| \geq 2$ so $\AA \nsubseteq \{a\}$ and $\AA \neq \emptyset$ and D wins the game.
	\end{itemize}
	
\end{proof}

The language we use encodes sets of \emph{the cumulative hierarchy}, defined as follows:
\begin{align*}
V_0 &:= \emptyset \\
V_{n+1} &:= \mathcal{P}(V_n).
\end{align*}
For each set in the cumulative hierarchy, we define a set of natural encodings. The encodings correspond to the different ways the set could be written down using only set brackets $\{$ and $\}$. To differentiate the encoded words from actual set notation, we will use parentheses $($ and $)$ instead. The encodings are defined as follows:
\begin{align*}
\enc(\emptyset) &:= \{()\} \\
\enc(X) &:= \{ (e_1\cdots e_n) \mid e_i \in \enc(x_i), x_1 < \cdots < x_n \text{ is a linear order of } X\}.
\end{align*}
A set has several encodings corresponding to different orders of the elements. For example, the set $V_2 = \{ \emptyset, \{\emptyset\}\}$ has the encodings $(()(()))$ and $((())())$. 

Let $\Sigma$ be the alphabet with $($ and $)$ and let $n \in \Nset$. We consider the following language:
\[
L_n = \bigcup\limits_{X \in V_{n+1}} enc(X).
\]

We first define $L_n$ in first-order logic with linear order $<$ and a unary predicate symbol $P$.

We define some auxiliary formulas. We interpret the predicate $P$ so that the left parentheses satisfy $P$ and the right parentheses do not. We use the formulas $L(x)$ and $R(x)$ to indicate this. We also define the formula $S(x,y)$ that says $y$ is the successor of $x$.
\[
L(x) := P(x), R(x) := \neg P(x), S(x,y) := x<y \land \neg \exists z (x<z<y)
\]

We will often want to say that the subword from position $x_1$ to $x_2$ encodes an instance of a set $X$. For easy readability of these kinds of statements, we adopt a flexible notation, where capital letters are used as shorthand for pairs of variables, that is to say $X := (x_1, x_2)$. Whenever possible, we shall use only the capital letters but in some cases we need the singular variables also.

We define the formulas $\pair_i(X)$ and $X =_i Y$ by mutual recursion. We additionally define formulas $X \in_i Y$, but since these only refer to the formula $\pair_i$, they are not essential in the recursion but rather shorthand to make the formulas more readable. The formula $\pair_i(X)$ says that $X$ correctly encodes a set in $V_i$ with no repetition. The formula $X \in_i Y$ assumes $Y$ encodes a set and says that $X$ encodes a set in $V_i$ and is an element of the set encoded by $Y$. Finally, the formula $X =_i Y$ assumes $X$ and $Y$ both encode sets in $V_i$ and says that these sets are the same. The definition by mutual recursion is as follows:
\begin{align*}
\pair_0(X) &:= L(x_1) \land R(x_2) \land S(x_1,x_2) \\
\pair_{i+1}(X) &:= x_1<x_2 \land L(x_1) \land R(x_2) \\
\land \forall u (x_1&<u<x_2 \rightarrow \exists v (x_1<v<x_2 \land (\pair_i(u,v) \lor \pair_i(v,u)))) \\
\land \forall A \forall B &((A \in_i X \land B \in_i X \land a_1 \neq b_1) \rightarrow A \neq_i B) \\
&\\
X \in_i Y &:= y_1 < x_1 < x_2 < y_2 \land \pair_i(X) \\
&\land \neg\exists U (y_1<u_1<x_1 \land x_2<u_2<y_2 \land \pair_i(U)) \\
&\\
X =_0 Y &:= \top \\
X =_{i+1} Y &:= \forall A(A \in_i X \rightarrow \exists B(B \in_i Y \land A =_i B))  \\
&\ \  \land \forall B(B \in_i Y \rightarrow \exists A(A \in_i X \land A =_i B)) 
\end{align*}
We use these auxiliary formulas to define the formula $\set_n$, which defines the language $L_n$. The formula $\set_n$ says that the first and last symbol of the word encode a set in $V_n$ with no repetition.
\begin{align*}
\set_{n} := &\exists X( \forall z (x_1 \leq z \land z \leq x_2) \land \pair_{n}(X))
\end{align*}
From the form of the formulas we see that $\size(\set_n) = \mathcal{O}(c^n)$ for some small constant $c$.\footnote{Numerical calculations performed with Maple seem to indicate $\size(\set_n) = \mathcal{O}(8^n)$.}

Now Proposition \ref{thm:nonelem} allows us to easily prove a non-elementary succinctness gap between $\FO$ and $\RE$. This gap already follows from the work of Stockmeyer \cite{stockmeyerthesis}. He found a similar gap between star-free expressions and $\RE$ and an elementary translation from $\FO$ to star-free expressions \cite{mcnaughton} leads to this result.

\begin{theorem}\label{cor:gap}
	$\FO(<,P)$ is non-elementarily more succinct than $\RE$ on words.
\end{theorem}
\begin{proof}
	The language $L_n$ is finite and $|L_n| \geq \tower(n)$. We have shown that $L_n$ can be defined in $\FO(<, P)$ via a formula exponential in $n$. However, if $k < \log(\tower(n)) = \tower(n-1)$, by Theorem \ref{thm:nonelem}, D has a winning strategy in the game $\game(k, L, \Sigma^* \setminus L)$. Thus, by Theorem \ref{thm:peli}, there is no $\RE$ that defines $L$ with size less than $\tower(n-1)$. 
\end{proof}

\section{Number of stars in RE over star-free}

We shift our attention from the overall size of regular expressions to only the number of stars. Star height famously gives a hierarchy in terms of expressive power for RE \cite{Hashiguchi88} and the corresponding result for GRE is a notorious open problem. For the number of stars, a full hierarchy can be trivially obtained already in star height one. On the other hand, for GRE, we have so far been unable to prove results of this nature due to the added complexity brought to the game with full use of complement. We present an interesting middle ground between RE and GRE we call RE over star-free. For these expressions, star-free, that is FO-definable, properties are combined using the operations of RE. For RE over star-free we show that the number of stars gives a hierarchy in terms of expressive power.

The aforementioned trivial hierarchy for RE is obtained via the expression $a_1^* \cup \cdots \cup a_n^*$ but we omit that proof since we prove the stronger hierarchy for RE over star-free expressions. The language we use is actually definable with $n$ stars already in RE but we show that even if we allow RE over star-free expressions, it still requires $n$ stars to define.

Let $\Sigma_n = \{a_1, \dots, a_n\}$ be a set of $n$ symbols. We consider the following $\Sigma_n$-language:
\[
	L_n := L\big(\bigcup_{i \in [n]} (a_1 \cup \cdots \cup a_{i-1} \cup a_i^2 \cup a_{i+1} \cup \cdots \cup a_n)^*\big)
\]
 In other words, for each word in $w \in L_n$, there is $i \in [n]$ such that every $a_i$-chain in $w$ has even length. We don't need the whole language $L_n$ for the game so we use a simple subset instead. For $k \in \Nset$ and $i \in [n]$, we define 
 \[
 	L_{n,k} := \{\ell_1, \dots, \ell_n\} = \{a_1^{2k + 1}\cdots a_i^{2k}\cdots a_n^{2k+1} \mid i \in [n]\}.	
 \]
 Each $\ell_i$ is a word that consists of a chain of each symbol $a_j$ in order. The chain of the specific symbol $a_i$ has even length and all other chains of $a_j$ have odd length.

\begin{theorem}
	Any RE over star-free expression $R_n$ with $L(R_n) = L_n$ has at least $n$ stars.
\end{theorem}
\begin{proof}
	Let $n \in \Nset$ and $k_0 \geq n$. We consider the languages $A_0 := L_{n,k_0}$ and $B_0 := \Sigma_n^* \setminus L_n$. We will show that D has a winning strategy for the game $\game(k_0, n-1, A_0, B_0)$. Since $A_0 \subseteq L_n$ and $B_0 = \Sigma_n^* \setminus L_n$, D then also has a winning strategy for the game $\game(k_0, n-1, L_n, \Sigma_n^* \setminus L_n)$. The number $k_0$ is arbitrary so by Theorem \ref{thm:gpeli} the claim follows.
	
	Let $(k, s, A, B)$ be a position in the game $\game(k_0, n-1, A_0, B_0)$. We will show that D can maintain the following conditions while a \starmove\ has not been made. We will also see that if a \starmove\ is made while the conditions hold, D gains a winning strategy. The conditions are:
	\begin{align*}
		&\text{\ \ There is $I\subseteq [n]$ such that} \\
		&\text{1. } |I|>s, \\
		&\text{2. for every $i \in I$ there is $w_i \in A$ and $u_i, v_i \in \Sigma_n^*$ s.t. $\ell_{i} = u_iw_iv_i$ and $(a_i)^{k+1}$ is a subword of $w_i$,} \\
		&\text{3. for every $r \in \Sigma_n^*$ if there are $i, j \in I$ with $u_irv_j \in B_0$, then $r \in B$.}
	\end{align*}
	
	Intuitively condition 2 says that in the position $(k, s, A, B)$, the set $A$ has some `descendants' $w_i$ of the original words $\ell_i$ in $A_0$. The words $u_i$ and $v_i$ are the parts that have been removed from $\ell_i$ via \catmove s to obtain $w_i$. The set $I$ contains the indices that still have descendants in play. Condition 1 states that the number of such indices is always larger than the star resource $s$. Finally condition 3 says that the set $B$ has versions of the original words in $B_0$ with some prefix $u_i$ and some suffix $v_j$ removed. 
	
	In the starting position $(k_0, n-1, A_0, B_0)$ the conditions hold with $I = [n]$ and for every $i \in I$, $w_i = \ell_i$ and $u_i = v_i = \epsilon$. We consider each possible move of S and show that in every case either the above conditions are maintained or D wins eventually by a winning strategy described in a previous lemma.
	\begin{itemize}
		\item \negmove: We must first check that while the conditions hold, a \negmove\ from S leads to a win for D. Let $i \in I$. By condition 2, the word $w_i$ has $(a_i)^{k+1}$ as a subword. Let $r$ be a word obtained from $w_i$ by adding one $a_i$ to this $a_i$-chain. Since $\ell_{i} = u_iw_iv_i$ and the $a_i$-chain in $\ell_i$ is even, we know the chain in $u_irv_i$ is odd. The chains of all other $a_j$ are odd in $\ell_i$ and thus also in $u_irv_i$ so $u_irv_i \in B_0$. By condition 3, we have $r \in B$. If S makes a \negmove, his star resource $s$ becomes $0$. In the following position $(k-1, 0, B, A)$, we have $r \in B$ and $w_i \in A$ and the two words only differ by the length of a chain with length more than $k-1$ so Lemma \ref{lemma:cup} gives D a winning strategy. This means that while the conditions hold, S can only attempt \cupmove s, \catmove s and \starmove s if he hopes to win.
		
		\item \cupmove: Let $A_1, A_2 \subseteq A$ be the subsets S chooses. For each $i \in I$,  $w_i \in A_1$ or $w_i \in A_2$. Let $I_1, I_2 \subseteq I$ be the sets of indices generated this way. Since $|I|>s$, we have $|I_1|>s_1$ or $|I_2|>s_2$. D chooses the position where this holds. Condition 2 still clearly holds and since $B$ remains unchanged in this move, so does condition 3.
		
		\item \catmove: Let $i \in I$ and let $(w_{i,1}, w_{i,2})$ be the split S chooses for $w_i$. Let $k_1 + k_2 + 1 = k$ and $s_1 + s_2 = s$ be the resource splits of S. Since $w_i$ has $(a_i)^{k+1}$ as a subword, $w_{i,1}$ has $(a_i)^{k_1+1}$ as a subword or $w_{i,2}$ has $(a_i)^{k_2+1}$ as a subword. We divide $I$ into subsets $I_1, I_2$ according to this condition. Since $|I| > s$, we have $|I_1| > s_1$ or $|I_2|>s_2$. Assume the former. Now condition 2 is satisfied for $w_{i,1}$ by letting $u_{i,1} := u_i$ and $v_{i,1} :=  w_{i,2}v_i$. For condition 3, let $u_{i,1}rv_{j,1} \in B_0$ for some $r \in \Sigma_n^*$ and $i, j \in I_1$. Now $u_irw_{j,2}v_j \in B_0$ so by condition 3 in the position before this move, $rw_{j,2} \in B$. For the split $(r, w_{j,2})$ of $rw_{j,2}$ S must choose $r$ to have a chance, since choosing $w_{j,2}$ would result in an identical word on both sides for the position $(k_2, s_2, A_2, B_2)$. So either D has a winning strategy by Lemma \ref{lem:sama} or $r \in B_1$ for every such $r$ and condition 3 holds for the position $(k_1, s_1, A_1, B_1)$ and D chooses this position. The case of $|I_2|>s_2$ is handled in the same way. 
		
		\item \starmove: S can only make this move if $1 \leq s < |I|$ so we have $i,j \in I$ with $i < j$. We will show that this is enough to give D a winning strategy if S makes a \starmove. Our aim is to show that a word of the form $(w_j)^{m_1}(w_i)^{m_2}$ is in $B$. We will use condition 3 to show this. Condition 3 requires a word of the form $u_irv_j$ to be in $B_0$ and words in $B_0$ have odd chains of all symbols $a_p$. Thus we begin by finding odd chains of all symbols in our words.
		
		Recall that by condition 2, there are $w_i \in A$ and $u_i, v_i \in \Sigma^*$ such that $\ell_i = u_iw_iv_i$ and $(a_i)^{k+1}$ is a subword of $w_i$. The same holds for $j$. Let $u \in \{u_i, u_j\}$ be the one of the two words with more odd chains of symbols. If they have the same number of odd chains, we choose, say, the longer word. Choose $v \in \{v_i, v_j\}$ the same way. Next, we will show that for each $p \in [n]$, at least one of the words $w_i$, $w_j$, $u$ and $v$ has an odd $a_p$-chain. 
		
		Recall that the words in $A_0$ have chains of symbols $a_p$ in order and only the $a_i$-chain in a word $\ell_i$ is even while all the others are odd. Furthermore, $\ell_i = u_iw_iv_i$ and $w_i$ has $(a_i)^{k+1}$ as a subword so all chains in $u_i$ are odd except possibly the last. Thus for each odd chain in $u_i$ there is also one of the same symbol in $u$ and the same goes for $u_j$. Similarly for each odd chain in $v_i$ or $v_j$ there is one in $v$. 
		
		We now show that for every $p \in [n]$ there is an odd chain in at least one of the words $w_i$, $w_j$, $u$ and $v$. First, let $p < i$. If there is an odd $a_p$-chain in $w_i$ we are done so let us assume there is not. Now the $a_p$-chain in $w_i$ is even (possibly empty) and since the chain in $u_iw_iv_i = \ell_i$ is odd, we know the one in $u_i$ is odd. As noted above, an odd chain in $u_i$ means there is also one in $u$. So in this case there is an odd $a_p$-chain in $w_i$ or $u$. The case $p > i$ is very similar and we obtain an $a_p$-chain in $w_i$ or $v$. Finally let $p = i$. Now $p < j$ so like above we obtain an odd $a_p$-chain in $w_j$ or~$u$. 
		
		We now have an odd chain of each $a_p$ among the words $w_i$, $w_j$, $u$ and $v$, but we still need to make sure the specific way we catenate these words does not remove the only odd chains of a symbol by merging them into an even one. Let $f(w)$ be the index of the first symbol of a word $w$ and $l(w)$ the index of the last. By condition 2 we have $f(w_i) \leq i \leq l(w_i)$. The same goes for $f(w_j) \leq j \leq l(w_j)$. We start with $w_jw_i$. By the above we obtain $f(w_i) \leq i < j \leq l(w_j)$ so this catenation cannot result in any merging of odd chains. Next we add $u$ to the left. If $l(u) = f(w_j)$ and both chains are odd, this merges the chains into an even one. Here we consider two cases. First, if $w_j$ is just an odd $a_{j}$-chain, then for some $m_1 \in \{1,2\}$ the $a_{j}$-chain in the word $u(w_j)^{m_1}w_i$ is odd. If $w_j$ has other symbols besides $a_{j}$, then the word $u(w_j)^2w_i$ has an odd $a_{f(w_j)}$-chain at the start of the second $w_j$. We have thus obtained $u(w_j)^{m_1}w_i$ with an odd chain of $a_{f(w_j)}$. We finally add $v$ to the right in a similar fashion. If $l(w_i) = f(v)$ and both chains are odd, we again consider the cases of $w_i$ being just an odd $a_i$-chain or a larger word and we obtain $m_2 \in \{1,2\}$ such that $u(w_j)^{m_1}(w_i)^{m_2}v$ has an odd chain of $a_{l(w_i)}$. 
		
		As the words $w_i$, $w_j$, $u$ and $v$ have an odd chain of each symbol and we have made sure the catenations did not lose any, our catenated word $u(w_j)^{m_1}(w_i)^{m_2}v$ is now in $B_0$. Since $u \in \{u_i, u_j\}$ and $v \in \{v_i, v_j\}$, by condition 3, $(w_j)^{m_1}(w_i)^{m_2} \in B$. 
		
		Let us finish by showing how this gives D a winning strategy after the \starmove\ in progress. S must give splits for $w_i$ and $w_j$ and every piece of these splits is in the left set of the following position, $A'$. S must also choose a piece of every split of $(w_j)^{m_1}(w_i)^{m_2}$ to add to the right set, $B'$. The split of $(w_j)^{m_1}(w_i)^{m_2}$ we are interested in is the one where each subword $w_i$ and $w_j$ is split according to the splits given by S for $w_i$ and $w_j$. For this split, S must choose one of the pieces already in $A'$ to also be in $B'$. Thus, in the following position $(k-1, s-1, A', B')$, there is an identical word on both sides and D has a winning strategy by Lemma \ref{lem:sama}. Thus if S makes a \starmove\ while the conditions hold, D eventually wins.
	\end{itemize}
\end{proof}

\section{Conclusion}

We have presented a formula size game for GRE, RE and a middle ground between these we call RE over star-free expressions. We used the $\RE$ version to reprove a non-elementary succinctness gap between $\FO$ and $\RE$ via a large finite language. For RE over star-free we showed that the number of stars gives a full hierarchy in terms of expressive power. As the astute reader has noted, we have not used the full $\GRE$ size game in this paper. This is due to the considerable combinatorial complexity of the game. A clear goal for further research is to find some handle on this complexity at least for some problems. A good first candidate is to prove that there is a star height one language that requires two stars to define via a GRE.

As noted in Remark \ref{remark}, the games can be modified to isolate different operations with different resources or counting the nesting depth of some operations instead of the number. This means that the games could naturally be used to investigate any problem having to do with bounds on operators such as the generalized star height problem.

\bibliographystyle{eptcs}
\bibliography{regamebib}
\end{document}